\documentclass[%
 reprint,
 amsmath,amssymb,
 aps,
pra,
]{revtex4-1}

\usepackage{graphicx}
\usepackage{dcolumn}
\usepackage{bm}
\usepackage{braket}
\usepackage{amsthm}

\newtheorem{theorem}{Theorem}

\begin{document}


\title{Quantum Error Correcting Code for Ternary Logic}

\author{Ritajit Majumdar$^1$} 
\email{majumdar.ritajit@gmail.com}
\author{Saikat Basu$^2$}
\author{Shibashis Ghosh$^2$}
\author{Susmita Sur-Kolay$^{1}$}
\email{ssk@isical.ac.in}

\affiliation{$^{1}$Advanced Computing \& Microelectronics Unit, Indian Statistical Institute, India\\
 $^{2}$A. K. Choudhury School of Information Technology, University of Calcutta, India}


\begin{abstract}
Ternary quantum systems are being studied because these provide more computational state space per unit of information, known as qutrit. A qutrit has three basis states, thus a qubit may be considered as a special case of a qutrit where the coefficient of one of the basis states is zero. Hence both $(2 \times 2)$-dimensional as well as $(3 \times 3)$-dimensional Pauli errors can occur on qutrits. In this paper, we (i) explore the possible $(2 \times 2)$-dimensional as well as $(3 \times 3)$-dimensional Pauli errors in qutrits and show that any pairwise bit swap error can be expressed as a linear combination of shift errors and phase errors, (ii) propose a new type of error called quantum superposition error and show its equivalence to arbitrary rotation, (iii) formulate a nine qutrit code which can correct a single error in a qutrit, and (iv) provide its stabilizer and circuit realization.
\end{abstract}

\pacs{Valid PACS appear here}
\maketitle


\section{\label{sec:intro}Introduction}

Quantum computers hold the promise of reducing the computational complexity of certain problems. However, quantum systems are highly sensitive to errors; even interaction with environment can cause a change of state. While descriptions of quantum algorithms, communication protocols, etc. assume the existence of closed quantum system, in reality such systems are open. Hence error correction is of utmost importance for quantum computation.

Multi-valued quantum computing based on $d$-dimensional quantum systems is gaining importance particularly in the field of quantum cryptography as it can represent a larger state space using less quantum resources. For example, in ion trap technology, if $d$ levels of the ion are used, then the number of qudits required can be reduced by a factor of $log_{2}d$ \cite{PhysRevA.62.052309}. In \cite{fan2008applications}, the author showed that  multi-valued quantum logic in some quantum algorithms outperforms its binary counterpart. Furthermore, recent researches on quantum algorithms have used lackadaisical quantum walk \cite{1751-8121-48-43-435304} with multi-dimensional coin to solve search problems faster than their classical counterpart \cite{10.1007/978-3-662-49192-8_31, wang2017adjustable}. The simplest higher dimensional quantum system is for $d = 3$.  Muthukrishnan and Stroud \cite{PhysRevA.62.052309} have implemented a few $(3 \times 3)$ quantum ternary logic gates with ion trap. 



A general ternary quantum state (qutrit) is represented as $\ket{\psi} = \alpha\ket{0} + \beta\ket{1} + \gamma\ket{2}$, where $\alpha, \beta, \gamma \in \mathbb{C}$ and $|\alpha|^2 + |\beta|^2 + |\gamma|^2 = 1$. For practical implementation of quantum ternary system in noisy environment, it is essential that one is able to correct errors in the system. In literature, 9-qubit \cite{PhysRevA.52.R2493}, 7-qubit \cite{PhysRevLett.77.793} and 5-qubit \cite{PhysRevLett.77.198} codes are available for error correction in qubits.  However, to the best of our knowledge, no error correcting code has been provided for ternary quantum systems. Previous attempts have been made for higher dimensional error correction \cite{Gottesman1999, PhysRevA.55.R839, PhysRevA.64.012310, PhysRevA.77.032309, PhysRevA.86.022308} but they have either considered $(2 \times 2)$-dimensional errors or $(d \times d)$-dimensional errors only and no explicit circuit has been presented.

{\it Main Contributions:} In this paper, we study error correction in qutrits considering both $(2 \times 2)$-dimensional as well as $(3 \times 3)$-dimensional errors. We have introduced a new type of error where a basis state is mapped to any arbitrary superposition state. We have called it \emph{quantum superposition error}. We have shown that such an error can be represented as an arbitrary diagonalizable rotation matrix. We have proposed an error correcting code that can correct these $(3 \times 3)$-dimensional as well as $(2 \times 2)$-dimensional errors affecting a qutrit.
Moreover, for qubits, a phase error in \{$\ket{+},\ket{-}$\} basis is analogous to a bit error in $\{\ket{0},\ket{1}\}$ basis, where $\ket{+} = \frac{1}{\sqrt{2}}(\ket{0} + \ket{1})$ and $\ket{-} = \frac{1}{\sqrt{2}}(\ket{0} - \ket{1})$. 
The \{$\ket{+},\ket{-}$\} basis is obtained by the Hadamard transform of $\{\ket{0},\ket{1}\}$ basis. However, in a qutrit, since there are three basis states, ternary QFT does not have similar properties as that of Hadamard in a qubit system. QFT acting on the states $\ket{0}, \ket{1}$, and $\ket{2}$ produces the states $\ket{+} = \frac{1}{\sqrt{3}}(\ket{0} + \ket{1} + \ket{2})$, $\ket{-} = \frac{1}{\sqrt{3}}(\ket{0} + \omega^2\ket{1} + \omega\ket{2})$ and $\ket{|} = \frac{1}{\sqrt{3}}(\ket{0} + \omega\ket{1} + \omega^2\ket{2})$ respectively. Hence the bit flip nature of \{$\ket{+},\ket{-}$\} basis for qubits does not translate directly in three-dimensional systems. This problem has not been addressed earlier for higher dimensional quantum error correction. We have proposed a mechanism to correct phase errors in qutrits. Moreover, we provide stabilizer and circuit realization of the error correction mechanism. We also show that similar to qubit system \cite{PhysRevLett.77.198} five qutrits are necessary to correct a single error in a qutrit. 


In Section II, we provide the possible types of pure state qutrit errors. We consider $(2 \times 2)$-dimensional pairwise swap error and $(3 \times 3)$-dimensional errors which we call \emph{shift} errors. Section III presents phase errors due to decoherence. We have also introduced the notion of a quantum superposition error, and shown its equivalence to an arbitrary phase error. In Section IV, we depict the basic error model  which captures all pure state errors, phase errors and arbitrary diagonalizable rotation. In Section V, we provide our proposed 9-qutrit code for correcting a single error in a qutrit. Section VI shows the stabilizer and circuit realizations of the error correction. We conclude in Section VII.

\section{\label{sec:bit}Bit errors in qutrits}
Our mathematical model of a qutrit, as mentioned in Section~\ref{sec:intro}, is in \{$\ket{0}, \ket{1}, \ket{2}$\} basis where $\ket{0} = \begin{pmatrix}
1 & 0 & 0
\end{pmatrix}^T$, $\ket{1} = \begin{pmatrix}
0 & 1 & 0
\end{pmatrix}^T$ and $\ket{2} = \begin{pmatrix}
0 & 0 & 1
\end{pmatrix}^T$. Hence a general error-free qutrit is of the form $\ket{\psi} = \alpha\ket{0} + \beta\ket{1} + \gamma\ket{2}$ where $\alpha, \beta, \gamma \in \mathbb{C}$ and $|\alpha|^2 + |\beta|^2 + |\gamma|^2 = 1$. We first consider bit \emph{pairwise swap} errors which are $(2 \times 2)$-dimensional. Then we discuss ternary bit errors which we call \emph{shift} errors.

\subsection{Ternary \emph{pairwise swap} errors}
It is possible to have a $(2 \times 2)$-dimensional error in a three-dimensional system. A pairwise pure state swap or bit flip error in a qubit \cite{nielsen2010quantum} is represented  by the Pauli matrix $\sigma_x = \begin{pmatrix}
0 & 1\\
1 & 0
\end{pmatrix}$. In a ternary system, there can be three types of $(2 \times 2)$-dimensional pairwise swap, namely $X_{01}, X_{12}, X_{20}$. A single qutrit pairwise swap error operates only on any two of the three basis states, i.e., the amplitudes of two out of the three basis states get swapped in the presence of such an error. The matrices corresponding to these errors are as follows:

\begin{itemize}
\item $X_{01}\ket{\psi} = \alpha\ket{1} + \beta\ket{0} + \gamma\ket{2}$.
\begin{center}
$X_{01} = \begin{pmatrix}
0 & 1 & 0\\
1 & 0 & 0\\
0 & 0 & 1
\end{pmatrix}$
\end{center}
\item $X_{12}\ket{\psi} = \alpha\ket{0} + \beta\ket{2} + \gamma\ket{1}$.
\begin{center}
$X_{12} = \begin{pmatrix}
1 & 0 & 0\\
0 & 0 & 1\\
0 & 1 & 0
\end{pmatrix}$
\end{center}
\item $X_{20}\ket{\psi} = \alpha\ket{2} + \beta\ket{1} + \gamma\ket{0}$.
\begin{center}
$X_{20} = \begin{pmatrix}
0 & 0 & 1\\
0 & 1 & 0\\
1 & 0 & 0
\end{pmatrix}$
\end{center}
\end{itemize}

These matrices are \emph{self-adjoint}. Hence if any of these errors occur, applying the same error matrix again can correct it.

\subsection{Ternary bit \emph{shift} errors}
In addition to bit flip errors, purely ternary errors are possible which we call (cyclic) \emph{shift errors}. Two types of shifts are possible - clockwise shift ($0 \rightarrow 1 \rightarrow 2$) and anticlockwise shift ($0 \leftarrow 1 \leftarrow 2$). The mathematical formulation of \emph{clockwise shift} ($X_1$) is $\ket{j} \xrightarrow{X_1} \ket{j + 1}$ mod 3 and that of \emph{anticlockwise shift} ($X_2$) is $\ket{j} \xrightarrow{X_2} \ket{j - 1}$ mod 3. It is interesting to note that the stabilizer proposed by Gottesman in \cite{Gottesman1999} for higher dimensional errors correspond to the \emph{clockwise shift} ($X_1$).

It can be easily verified that the respective matrices corresponding to errors $X_1$ and $X_2$ are -

\begin{center}
	\begin{tabular}{ c  c  c}
		$X_1 = \begin{pmatrix}
		0 & 0 & 1\\
		1 & 0 & 0\\
		0 & 1 & 0
		\end{pmatrix}$
		&
		\hspace{4mm}
		&
		$X_2 = \begin{pmatrix}
		0 & 1 & 0\\
		0 & 0 & 1\\
		1 & 0 & 0
		\end{pmatrix}$
	\end{tabular}
\end{center}

The action of these cyclic shift errors on the error-free state $\ket{\psi}$ can be mathematically represented as\\

$X_1\ket{\psi} = \alpha\ket{1} + \beta\ket{2} + \gamma\ket{0}$.\\

$X_2\ket{\psi} = \alpha\ket{2} + \beta\ket{0} + \gamma\ket{1}$.\\

The matrices for \emph{shift} errors are not self-adjoint. However, each type of shift error occurring twice in succession produces the other type of shift error, i.e., $X_2 = X_1^2$ and $X_1 = X_2^2$; further $X_1^{-1} = X_2$ and $X_2^{-1} = X_1$. Thus to correct an $X_1$ error, one can apply $X_2$ and vice-versa. 

It can be checked that any combination of these five types of errors (pairwise swap and shift) results in one of these five errors or identity. Thus these five errors exhaust the list of possible bit errors in a qutrit.

\section{\label{sec:phase}Phase errors in qutrits}
\subsection{Arbitrary rotation and phase}
Qutrits are unit vectors in three-dimensional Hilbert Space $\mathcal{H}^{\otimes 3}$. On interaction with the environment, a qutrit can undergo rotation by some arbitrary angles, where the basis states $\ket{1}$ and $\ket{2}$ may incur different phase errors given by $e^{i\theta}$ and $e^{i\phi}$ respectively. Such an error changes the error- free state $\ket{\psi}$ as
\begin{center}
    $\alpha\ket{0} + \beta\ket{1} + \gamma\ket{2} \rightarrow  \alpha\ket{0} + \beta e^{i\theta}\ket{1} + \gamma e^{i\phi}\ket{2}$.
\end{center}

The corresponding error operator is denoted as
\begin{center}
$R_{\theta\phi} = \begin{pmatrix}
1 & 0 & 0\\
0 & e^{i\theta} & 0\\
0 & 0 & e^{i\phi}
\end{pmatrix}$.
\end{center}

Using the formula $e^{\pm i\theta} = cos\theta \pm i sin\theta$, the error matrix can be represented upto a global phase of $e^{i\frac{\theta + \phi}{2}}$, as

\begin{center}
[$cos\frac{\theta}{2}\begin{pmatrix}
1 & 0 & 0\\
0 & 1 & 0\\
0 & 0 & 1
\end{pmatrix} - isin\frac{\theta}{2}\begin{pmatrix}
1 & 0 & 0\\
0 & -1 & 0\\
0 & 0 & 1
\end{pmatrix}$]
$\cdot $ [$cos\frac{\phi}{2}\begin{pmatrix}
1 & 0 & 0\\
0 & 1 & 0\\
0 & 0 & 1
\end{pmatrix} - isin\frac{\phi}{2}\begin{pmatrix}
1 & 0 & 0\\
0 & 1 & 0\\
0 & 0 & -1
\end{pmatrix}$]\\
\begin{equation}\label{eq1}
\begin{split}
= cos\frac{\theta}{2}cos\frac{\phi}{2}\mathbb{I} - isin\frac{\theta}{2}cos\frac{\phi}{2}Z_1 - icos\frac{\theta}{2}sin\frac{\phi}{2}Z_2 \\
- sin\frac{\theta}{2}sin\frac{\phi}{2}Z_{12}.
\end{split}
\end{equation}
\end{center}
 where

\begin{itemize}
\item $Z_1 = \begin{pmatrix}
1 & 0 & 0\\
0 & -1 & 0\\
0 & 0 & 1
\end{pmatrix}$\\
$\alpha\ket{0} + \beta\ket{1} + \gamma\ket{2} \xrightarrow{Z_1} \alpha\ket{0} - \beta\ket{1} + \gamma\ket{2}$.
\item $Z_2 = \begin{pmatrix}
1 & 0 & 0\\
0 & 1 & 0\\
0 & 0 & -1
\end{pmatrix}$\\
$\alpha\ket{0} + \beta\ket{1} + \gamma\ket{2} \xrightarrow{Z_2} \alpha\ket{0} + \beta\ket{1} - \gamma\ket{2}$.
\item $Z_{12} = \begin{pmatrix}
1 & 0 & 0\\
0 & -1 & 0\\
0 & 0 & -1
\end{pmatrix}$ = $Z_1Z_2 = Z_2Z_1$\\
$\alpha\ket{0} + \beta\ket{1} + \gamma\ket{2} \xrightarrow{Z_{12}} \alpha\ket{0} - \beta\ket{1} - \gamma\ket{2}$.
\end{itemize}

Correction of phase error is not so obvious in ternary systems as in binary systems. This is because phase flip error has the nature of bit flip error in the Hadamard basis ($\ket{+} = \frac{1}{\sqrt{2}}(\ket{0} + \ket{1})$, $\ket{-} = \frac{1}{\sqrt{2}}(\ket{0} - \ket{1})$). However, the natural extension of Hadamard transform, i.e., Quantum Fourier Transform (QFT) does not behave in a similar way for qutrits. We address this problem and propose a solution in Section~\ref{sec:stab}.

The phase of the three basis states are defined \cite{Gottesman1999} in terms of $\omega$, the cube root of unity, i.e., $\omega^3 = 1$, and $1 + \omega + \omega^2 = 0$. The two phase error matrices for $\ket{0}$ are given by

\begin{center}
	\begin{tabular}{ c  c  c}
		$R_1 = \begin{pmatrix}
		1 & 0 & 0\\
		0 & \omega & 0\\
		0 & 0 & \omega^2
		\end{pmatrix}$
		&
		\hspace{4mm}
		&
		$R_2 = \begin{pmatrix}
		1 & 0 & 0\\
		0 & \omega^2 & 0\\
		0 & 0 & \omega
		\end{pmatrix}$.
	\end{tabular}
\end{center}



It can be seen that the matrices $R_1$ and $R_2$ are special cases of the rotation matrix $R_{\theta\phi}$. So they are not independent errors. However, in Section~\ref{sec:stab}, we have used $\omega$ and $\omega^2$ phases to show that a binary bit pairwise swap error can be represented as a linear combination of ternary errors. Hence, we explicitly show their correction in Section~\ref{sec:stab}.

In the following part, we consider a special type of rotation, where each basis state is knocked into superposition of two or three basis states. We call it quantum superposition error, and describe it in details in the following subsection.
 
\subsection{Quantum superposition error}
We consider a special quantum error where the decoherence \emph{knocks} a basis state into a superposition of two (or three) basis states, i.e.,

\begin{center}
$\ket{i} \rightarrow \alpha\ket{j} + \beta\ket{k}$ where \{$i,j,k$\} $\in \{0,1,2\}$ and $j \ne k$.
\end{center}

The criterion $j \ne k$ ensures that a qutrit is not mapped to a qubit. Such an error can induce  the following:
\begin{eqnarray*}
\ket{0} & \rightarrow & \alpha_0\ket{j_0} + \beta_0\ket{k_0}.\\
\ket{1} & \rightarrow & \alpha_1\ket{j_1} + \beta_1\ket{k_1}.\\
\ket{2} & \rightarrow & \alpha_2\ket{j_2} + \beta_2\ket{k_2}.
\end{eqnarray*}

where \{$j_0, k_0, j_1, k_1, j_2, k_2$\} $\in$ \{0, 1, 2\} and the criterion ($j_l \ne k_l$ for $l\in \{0,1,2\}$) holds in each case. We call this \emph{quantum superposition error}.

\begin{theorem}
A quantum superposition error is equivalent to arbitrary diagonalizable rotation.
\end{theorem}

\begin{proof}
Since the error free state was $\ket{\psi} = \alpha\ket{0} + \beta\ket{1} + \gamma\ket{2}$, the error state can be written as $\ket{\phi} = a'\ket{0} + b'\ket{1} + c'\ket{2}$ by summing up the new amplitudes, where $|a'|^2 + |b'|^2 + |c'|^2 = 1$ and $a', b', c' \in \mathbb{C}$.

Since $|e^{i\lambda}| = 1$ for any $\lambda \in \mathbb{R}$, the new amplitudes $a'$, $b'$, $c'$ can be represented as $a' = e^{ia}\alpha$, $b' = e^{ib}\beta$ and $c' = e^{ic}\gamma$ for some $a, b, c \in \mathbb{R}$. So the state $\ket{\phi}$ can be written as $\ket{\phi} = e^{ia}\alpha\ket{0} + e^{ib}\beta\ket{1} + e^{ic}\gamma\ket{2}$
$\approx \alpha\ket{0} + e^{i(b - a)}\beta\ket{1} + e^{i(c - a)}\gamma\ket{2}$ up to a global phase of $e^{ia}$. Putting $b - a = \theta$ and $c - a = \phi$, this error can be represented by the error matrix $R_{\theta\phi}$ and can be corrected in a similar way.

The proof is similar for the case where a basis state is mapped to a superposition of three basis states i.e. $a_0\ket{0} + a_1\ket{1} + a_2\ket{2}$.
\end{proof}

It has already been shown that $Z_{12} = Z_1Z_2$. One can check that any combination of $Z_1, Z_2$ and $Z_{12}$ results in one of the three errors $Z_1, Z_2$, $Z_{12}$ or identity. Moreover, arbitrary phase and quantum superposition error can also be represented as combinations of these three. Hence, these three errors exhaust the list of possible phase errors in a qutrit.

\section{Error model for qutrits}
\label{sec:model}
The errors considered in Sections~\ref{sec:bit} and ~\ref{sec:phase} can be summed up using the error model proposed in \cite{PhysRevA.55.R839}. Any error $E$ acting on the qutrit is a $3 \times 3$ complex unitary matrix, $E \in \mathbb{C}^{3 \times 3}$. Such a matrix can be written as a linear combination of bit errors ($X$), phase errors ($Z$) \cite{PhysRevA.55.R839} and their product ($Y = iZX$). The general error model considered in this paper is

\begin{equation} \label{eq:model}
E = a\mathbb{I}_3 + \sum_{i = 1}^2b_iZ_i + \sum_{\substack{m, n = 0 \\ m \neq n}}^2(c_{mn}X_{mn} + \sum_{j = 1}^2d_{mnj}Y_{mnj})
\end{equation}

where $a, b, c_{mn}, d_{mnj} \in \mathbb{C}$ are constants. $\mathbb{I}_3$ is the $3 \times 3$ identity operator and $Y_{mnj} = iZ_jX_{mn}$ takes into account when both bit error ($X$) and phase error ($Z$) occur simultaneously. This error model tackles both binary as well as ternary bit and phase errors. Furthermore, it can encapsulate an arbitrary diagonalizable rotation error whose matrix representation is a diagonal matrix ($R_{\theta\phi}$) with different phases for the basis states. We consider this error model since it can capture any such error $E \in \mathbb{C}^{3 \times 3}$.

\section{\label{sec:code}Error correcting code for qutrits}

We propose a nine-qutrit error correcting code. For error correction, we encode the information of a single qutrit $\ket{\psi} = \alpha\ket{0} + \beta\ket{1} + \gamma\ket{2}$ into nine qutrits where the logical qutrit is denoted as \\
$\ket{\psi}_L = \alpha\ket{0}_L + \beta\ket{1}_L + \gamma\ket{2}_L$.

The logical $\ket{0}_L$, $\ket{1}_L$ and $\ket{2}_L$ are as follows:

\begin{eqnarray*}
\ket{0}_L& = & \frac{1}{3\sqrt{3}}(\ket{000} + \ket{111} + \ket{222})(\ket{000} + \ket{111} + \ket{222})\\
& &(\ket{000} + \ket{111} + \ket{222}).\\
\ket{1}_L& = & \frac{1}{3\sqrt{3}}(\ket{000} + \omega\ket{111} + \omega^2\ket{222})\\
& &(\ket{000} + \omega\ket{111} + \omega^2\ket{222})\\
& &(\ket{000} + \omega\ket{111} + \omega^2\ket{222}).\\
\ket{2}_L& = & \frac{1}{3\sqrt{3}}(\ket{000} + \omega^2\ket{111} + \omega\ket{222})\\
& &(\ket{000} + \omega^2\ket{111} + \omega\ket{222})\\
& &(\ket{000} + \omega^2\ket{111} + \omega\ket{222}).
\end{eqnarray*}

It is easy to check that $\ket{0}_L$, $\ket{1}_L$ and $\ket{2}_L$ are orthogonal to each other.

In order to correct the error in Eq.(\ref{eq:model}), ancilla state(s) $\ket{\zeta}$ is entangled with the system. Finally, the ancilla state(s) is measured which gives a classical outcome called \emph{error syndrome}. The error syndrome denotes the type of error that has occurred. The resultant state after entanglement of the ancilla state(s) $\ket{\zeta}$ is a superposition of the form
\begin{center}
$cos\frac{\theta}{2}cos\frac{\phi}{2}\mathbb{I}\ket{\psi}\ket{\zeta_{I}} - isin\frac{\theta}{2}cos\frac{\phi}{2}Z_1\ket{\psi}\ket{\zeta_{Z_{1}}} - icos\frac{\theta}{2}sin\frac{\phi}{2}Z_2\ket{\psi}\ket{\zeta_{Z_{2}}} - sin\frac{\theta}{2}sin\frac{\phi}{2}Z_{12}\ket{\psi}\ket{\zeta_{Z_{12}}}$
\end{center}

where $\ket{\zeta_i}$ indicates the ancilla qubit with $i$-th error syndrome. Upon measurement of the ancilla qubits, the superposition collapses. If the ancilla state collapses with $i$-th syndrome, then the encoded state also collapses with the $i$-th error on the system.

\section{\label{sec:stab}Stabilizer and circuit realization}

We have considered that a ternary quantum system can be perturbed by both binary (pairwise swap of basis states) and ternary (cyclic shift of basis states) errors. However, any pairwise swap error can be written as a linear combination of ternary shift and phase errors. For example, the error operator $X_{12}$ can be written (upto a normalization factor) as a linear combination of shift and phase operators as follows:

\begin{eqnarray*}
    \begin{pmatrix}
    1 & 0 & 0\\
    0 & 0 & 1\\
    0 & 1 & 0
    \end{pmatrix}& = & \begin{pmatrix}
    1 & 0 & 0\\
    0 & 1 & 0\\
    0 & 0 & 1
    \end{pmatrix} + \begin{pmatrix}
    1 & 0 & 0\\
    0 & \omega & 0\\
    0 & 0 & \omega^2
    \end{pmatrix} + \begin{pmatrix}
    1 & 0 & 0\\
    0 & \omega^2 & 0\\
    0 & 0 & \omega
    \end{pmatrix}\\
    & + & \begin{pmatrix}
    0 & 1 & 0\\
    0 & 0 & 1\\
    1 & 0 & 0
    \end{pmatrix} + \begin{pmatrix}
    0 & \omega & 0\\
    0 & 0 & 1\\
    \omega^2 & 0 & 0
    \end{pmatrix} + \begin{pmatrix}
    0 & \omega^2 & 0\\
    0 & 0 & 1\\
    \omega & 0 & 0
    \end{pmatrix}\\
    & + & \begin{pmatrix}
    0 & 0 & 1\\
    1 & 0 & 0\\
    0 & 1 & 0
    \end{pmatrix} + \begin{pmatrix}
    0 & 0 & \omega\\
    1 & 0 & 0\\
    0 & \omega^2 & 0
    \end{pmatrix} + \begin{pmatrix}
    0 & 0 & \omega^2\\
    1 & 0 & 0\\
    0 & \omega & 0
    \end{pmatrix}.
\end{eqnarray*}

Note that the error operators with angles $\omega$ and $\omega^2$ can be considered as $R_{\theta\phi}$ error for particular values of $\theta$ and $\phi$. Hence, a code which can correct shift and phase errors can also correct pairwise swap errors occurring on qutrits. Henceforth, we shall focus only on the correction of bit shift errors and phase errors.

\subsection{Stabilizer structure for error detection}

A set of operators $M_1, M_2, \hdots, M_k$ are called stabilizers of a quantum state $\ket{\psi}$ if $M_i\ket{\psi} = \ket{\psi}$ for all $i$. If the state $\ket{\psi}$, upon incurring some error, is changed to $\ket{\psi_E}$, then at least for one of the stabilizers $M_i$, the error state will be a -1 eigenstate, i.e., $M_i\ket{\psi_E} = -\ket{\psi_E}$ for at least one $i$. For ternary systems, the error states may be $\omega$-eigenstate or $\omega^2$-eigenstate of the stabilizers.

Gottesman defined stabilizers for higher dimensional spin systems as 

\begin{center}
    $X\ket{j} = \ket{j+1}$ mod d \hspace*{0.7cm} $R\ket{j} = \omega^j\ket{j}$
\end{center}

where $d$ is the dimension of the quantum state. For qutrit systems, $d = 3$. The stabilizers for detection of bit and phase errors are

\begin{center}
    $RRRIIIIII$, $IIIRRRIII$, $IIIIIIRRR$\\
    $XXXIIIIII$, $IIIXXXIII$, $IIIIIIXXX$.
\end{center}

The first three stabilizers, including $R$ and identity ($I$) only check shift errors. The last three stabilizers, including $X$ and identity only check the errors with $\omega$, $\omega^2$ phase. We show in Table~\ref{tab:stab} the eigenvalues corresponding to stabilizers $RRRIIIIII$ and $XXXIIIIII$ for different error states. The actions of the other stabilizers are similar. However, the stabilizers for shift errors can only detect the presence and the type of error, but not their location. Hence, a second step is required to find the location of the errors.

\begin{table}[]
    \centering
    \caption{Stabilizers for ternary errors}
    \begin{tabular}{| c | c | c |}
    \hline
        Error state & $RRRIIIIII$ & $XXXIIIIII$ \\
        \hline
        $\ket{000} + \ket{111} + \ket{222}$ & $1$ & $1$\\
        \hline
        $\ket{200} + \ket{011} + \ket{122}$ & $\omega^2$ & \\
        $\ket{020} + \ket{101} + \ket{212}$ & $\omega^2$ & \\
        $\ket{002} + \ket{110} + \ket{221}$ & $\omega^2$ & \\
        \hline
        $\ket{100} + \ket{211} + \ket{022}$ & $\omega$ & \\
        $\ket{010} + \ket{121} + \ket{202}$ & $\omega$ & \\
        $\ket{001} + \ket{112} + \ket{220}$ & $\omega$ & \\
        \hline
        $\ket{000} + \omega\ket{111} + \omega^2\ket{222}$ & & $\omega^2$\\
        $\ket{000} + \omega^2\ket{111} + \omega\ket{222}$ & & $\omega$\\
        \hline
    \end{tabular}
    \label{tab:stab}
\end{table}

\begin{figure}
    \centering
    \includegraphics[scale=0.9]{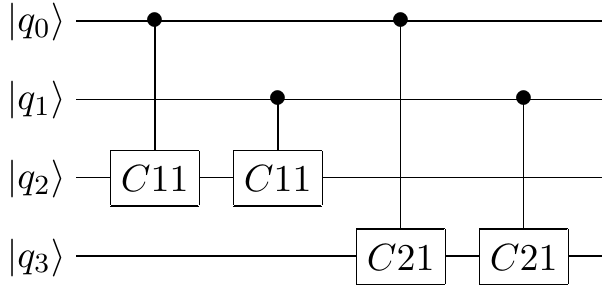}
    \caption{Circuit to compare two qutrits}
    \label{fig:block}
\end{figure}

\subsection{Error correction circuit}

After applying the stabilizer $RRRIIIIII$, if the eigenvalue is $\omega^2$, it implies that $X_2$ error has occurred in any one of the first three qutrits. However, it cannot specify the qutrit which has incurred the error. Fig~\ref{fig:block} shows the circuit which checks whether two qutrits are in the same state. In this circuit, we consider the ancilla states ($\ket{q_2}$, $\ket{q_3}$) to be \emph{qubits} only. Using qutrits as ancilla does not hamper the error correction procedure, but qutrits are not necessary. The truth table of this circuit is shown in Table~\ref{tab:truth}. The states of the ancilla qubits comprise the error syndrome. When the syndrome is 00, it implies that both the qutrits are in the same state. Any other error syndrome indicates that the qutrits are in different states.

\begin{table}[]
    \centering
    \caption{Truth table for the circuit in Fig~\ref{fig:block}}
    \begin{tabular}{| c | c | c | c |}
    \hline
        $\ket{q_0}$ & $\ket{q_1}$ & $\ket{q_2}$ & $\ket{q_3}$ \\
        \hline
        0 & 0 & 0 & 0\\
        0 & 1 & 1 & 0\\
        0 & 2 & 0 & 1\\
        1 & 0 & 1 & 0\\
        1 & 1 & 0 & 0\\
        1 & 2 & 1 & 1\\
        2 & 0 & 0 & 1\\
        2 & 1 & 1 & 1\\
        2 & 2 & 0 & 0\\
        \hline
    \end{tabular}
    \label{tab:truth}
\end{table}

The gates $C11$ and $C21$ are defined as follows:\\

$C_{11}:$ {\texttt{if}} (control = 1) {\texttt{then}}  target = target + 1\\

$C_{21}:$ {\texttt{if}}  (control = 2) {\texttt{then}}  target = target + 1

where the addition is modulo 3. \\
$C21$ is one of the basic $(3 \times 3)$ ternary quantum gates proposed and implemented in ion-trap by Muthukrishnan and Stroud \cite{PhysRevA.62.052309}. The gates proposed in \cite{PhysRevA.62.052309} are often called MS gates. Two other types of MS gates, namely MS+1 and MS+2, are defined as

\begin{center}
$MS+i\ket{j} = \ket{j+i}$ mod 3.
\end{center}

The implementation of $C11$ gate using MS gates is shown in Fig~\ref{fig:dec}.

\begin{figure}
    \centering
    \includegraphics{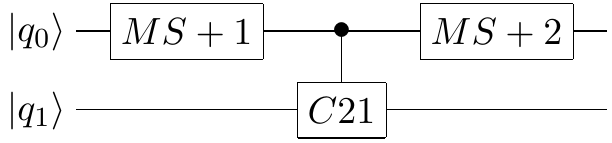}
    \caption{Realising the C11 gate with MS gates}
    \label{fig:dec}
\end{figure}

In Fig~\ref{fig:test} we show the circuit which checks whether the three qutrits are in the same state or not. In the Figure, $\ket{q_3}$ to $\ket{q_6}$ are ancilla \emph{qubits} initialized to $\ket{0}$. Following the truth table from Table~\ref{tab:truth}, if all the qutrits are in the same state, then the syndrome is 000. Otherwise one of the three bits will differ, which indicates the qutrit which has incurred error.

\begin{figure}
    \centering
    \includegraphics[scale=0.7]{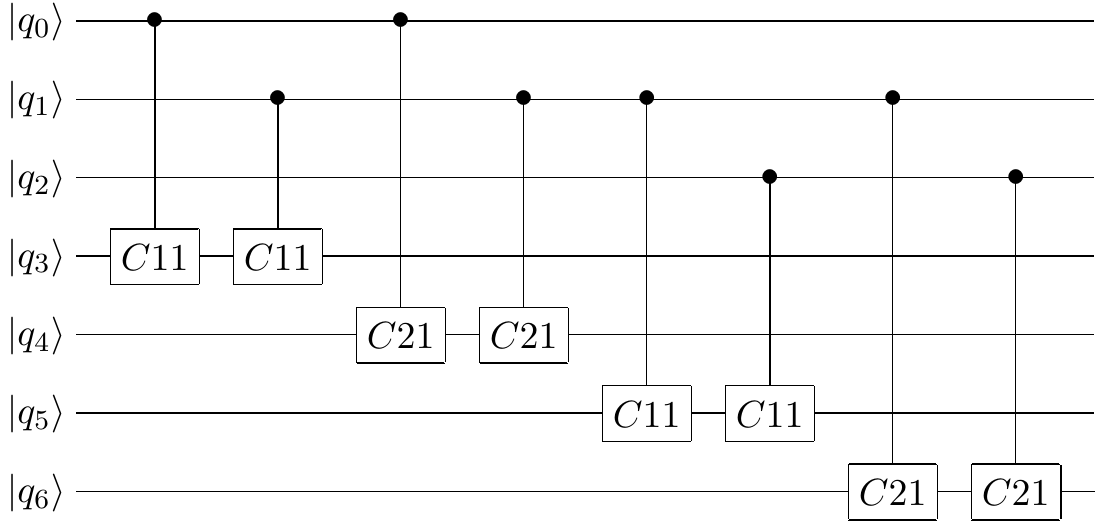}
    \caption{Circuit for qutrit error correction}
    \label{fig:test}
\end{figure}

Next, we address the problem of correction of phase errors ($Z_1$, $Z_2$ and $Z_{12}$). The \{$\ket{+}$, $\ket{-}$, $\ket{|}$\} basis for qutrits is equivalent to Hadamard basis for qubits.
\begin{eqnarray*}
\ket{+} &=& \frac{1}{\sqrt{3}}(\ket{0} + \ket{1} + \ket{2}).\\
\ket{-} &=& \frac{1}{\sqrt{3}}(\ket{0} + \omega^2\ket{1} + \omega\ket{2}).\\
\ket{|} &=& \frac{1}{\sqrt{3}}(\ket{0} + \omega\ket{1} + \omega^2\ket{2}).
\end{eqnarray*}

It is clear that the bit flip nature of Hadamard basis does not hold true in a three-dimensional system. In order to correct a phase error in qutrit, we consider the three unitary matrices as proposed in \cite{di2011elementary}:
\begin{itemize}
\item $H^{01} = \frac{1}{\sqrt{2}}\begin{pmatrix} 
1& 1 & 0\\
1 & -1 & 0\\
0 & 0 & \sqrt{2}
\end{pmatrix}$.
\item $H^{12} = \frac{1}{\sqrt{2}}\begin{pmatrix} 
\sqrt{2} & 0 & 0\\
0 & 1 & 1\\
0 & 1 & -1
\end{pmatrix}$.
\item $H^{20} = \frac{1}{\sqrt{2}}\begin{pmatrix} 
1 & 0 & 1\\
0 & \sqrt{2} & 0\\
1 & 0 & -1
\end{pmatrix}$.
\end{itemize}

These matrices are similar to Hadamard operation on two states while the third state is kept unchanged. Applying $H^{01}$ on these states changes $\ket{0}$ and $\ket{1}$ to $\ket{+}$ and $\ket{-}$ respectively, while the state $\ket{2}$ remains unchanged. Hence, if there is a phase error on $\ket{1}$ with respect to $\ket{0}$, then it can be easily detected as it will flip the states $\ket{+}$ and $\ket{-}$. Similarly, by applying $H^{20}$ any phase error between $\ket{0}$ and $\ket{2}$ can be detected. Since $Z_{12} = Z_1Z_2$, correcting $Z_1$ and $Z_2$ one after the another corrects $Z_{12}$. 

\subsection{Performance analysis}
The code proposed in this paper is a repetition code, where each of the logical qutrits ($\ket{0}_L$, $\ket{1}_L$, $\ket{2}_L$) are arranged in three blocks of three qutrits each. This approach is similar to Shor code \cite{PhysRevA.52.R2493} for qubits. If $p$ is the probability that a single qutrit is affected by decoherence, then the probability that none of the nine qutrits decohere is $(1-p)^9$. This code fails if more than one qutrit incur error. The probability that at least two qutrits have error is $1 - (1-p)^9 - 9p(1-p)^8 = 1 - (1+8p)(1-p)^8 \approx 36p^2$. Hence, when the error probability is less than $\frac{1}{36}$, this technique provides an improved method to preserve the coherence of the qutrits. The performance of our proposed code is in accordance to Shor code. However, this apparent similarity to Shor Code vanishes in the error correction process. Unlike Shor Code, the error correction is two fold - in the first step error is detected, and then its location is identified.

In the error model chosen for qutrits in this paper, there are two independent bit errors ($X_1$, $X_2$) and two independent phase errors ($Z_1$, $Z_2$) and their product ($Z_{12}$). So there are three phase errors. Hence, according to Eq.(\ref{eq:model}), there can be six $Y_{mnj} = iZ_jX_{mn}$ errors. Reliable detection of these eleven errors and the error free state demands each error state and the error free state to be in different orthogonal subspaces. An $n$-qutrit quantum system resides in a $3^n$-dimensional Hilbert Space. So, in an $n$-qutrit code, the number of orthogonal subspaces required cannot be more than $3^n$. To accommodate all these eleven errors and the error free state in separate orthogonal subspaces for each of the three logical qutrits in an $n$-qutrit code, Eq.(\ref{eq:test}) should be satisfied.

\begin{equation}
    3(11n+1) \leq 3^n.
    \label{eq:test}
\end{equation}

The minimum value of $n$ for which this inequality is satisfied is five. So five qutrits are necessary for correcting a single error in a qutrit. The similar bound was achieved by Laflamme et al \cite{PhysRevLett.77.198} for qubits. Our proposed code uses nine qutrits.

\section{Conclusion}
In this paper, we have proposed a nine qutrit code which can correct a single binary (pairwise swap, $Z_1$ and $Z_2$ phase and their combination) or ternary (shift, $Z_{12}$ and their combination) error. This code is a repetition code, where each of the logical qutrits ($\ket{0}_L$, $\ket{1}_L$, $\ket{2}_L$) are arranged in three blocks of three qutrits each. To the best of our knowledge, this is the first error correcting code for qutrit system which takes into account both binary and ternary errors. We have also addressed the issue of correcting phase errors in ternary systems since the bit flip nature of phase errors in Hadamard basis does not translate directly in ternary QFT basis. The performance of this code is similar to the nine-qubit code proposed by Shor, however, its stabilizer structure is not. We have shown that this code is not optimal in the number of qutrits for a ternary quantum system. Further studies can be performed to find the five-qutrit error correcting code. Moreover, it is worthwhile to study whether the two step error correction procedure used in this paper can be avoided and a single stabilizer structure can be provided for error correction.
\small 
\section*{Acknowledgement}
We gratefully acknowledge fruitful discussions with Prof. Peter Shor, MIT and Prof. Guruprasad Kar, Indian Statistical Institute, Kolkata. Ritajit Majumdar would like to acknowledge DST INSPIRE Fellowship DST/INSPIRE/03/2016/002097 for supporting the research.

\bibliographystyle{apalike}
\bibliography{main}

\begin{thebibliography}{}

\bibitem[Cafaro et~al., 2012]{PhysRevA.86.022308}
Cafaro, C., Maiolini, F., and Mancini, S. (2012).
\newblock Quantum stabilizer codes embedding qubits into qudits.
\newblock {\em Phys. Rev. A}, 86:022308.

\bibitem[Chau, 1997]{PhysRevA.55.R839}
Chau, H.~F. (1997).
\newblock Correcting quantum errors in higher spin systems.
\newblock {\em Phys. Rev. A}, 55:R839--R841.

\bibitem[Di and Wei, 2011]{di2011elementary}
Di, Y.-M. and Wei, H.-R. (2011).
\newblock Elementary gates for ternary quantum logic circuit.
\newblock {\em arXiv preprint arXiv:1105.5485}.

\bibitem[Fan, 2008]{fan2008applications}
Fan, Y. (2008).
\newblock Applications of multi-valued quantum algorithms.
\newblock {\em arXiv preprint arXiv:0809.0932}.

\bibitem[Gottesman, 1999]{Gottesman1999}
Gottesman, D. (1999).
\newblock {\em Fault-Tolerant Quantum Computation with Higher-Dimensional
  Systems}, pages 302--313.
\newblock Springer Berlin Heidelberg, Berlin, Heidelberg.

\bibitem[Gottesman et~al., 2001]{PhysRevA.64.012310}
Gottesman, D., Kitaev, A., and Preskill, J. (2001).
\newblock Encoding a qubit in an oscillator.
\newblock {\em Phys. Rev. A}, 64:012310.

\bibitem[Laflamme et~al., 1996]{PhysRevLett.77.198}
Laflamme, R., Miquel, C., Paz, J.~P., and Zurek, W.~H. (1996).
\newblock Perfect quantum error correcting code.
\newblock {\em Phys. Rev. Lett.}, 77:198--201.

\bibitem[Muthukrishnan and Stroud, 2000]{PhysRevA.62.052309}
Muthukrishnan, A. and Stroud, C.~R. (2000).
\newblock Multivalued logic gates for quantum computation.
\newblock {\em Phys. Rev. A}, 62:052309.

\bibitem[Nahimovs and Rivosh, 2016]{10.1007/978-3-662-49192-8_31}
Nahimovs, N. and Rivosh, A. (2016).
\newblock Quantum walks on two-dimensional grids with multiple marked
  locations.
\newblock In Freivalds, R.~M., Engels, G., and Catania, B., editors, {\em
  SOFSEM 2016: Theory and Practice of Computer Science}, pages 381--391,
  Berlin, Heidelberg. Springer Berlin Heidelberg.

\bibitem[Nielsen and Chuang, 2010]{nielsen2010quantum}
Nielsen, M.~A. and Chuang, I.~L. (2010).
\newblock {\em Quantum computation and quantum information}.
\newblock Cambridge university press.

\bibitem[Pirandola et~al., 2008]{PhysRevA.77.032309}
Pirandola, S., Mancini, S., Braunstein, S.~L., and Vitali, D. (2008).
\newblock Minimal qudit code for a qubit in the phase-damping channel.
\newblock {\em Phys. Rev. A}, 77:032309.

\bibitem[Shor, 1995]{PhysRevA.52.R2493}
Shor, P.~W. (1995).
\newblock Scheme for reducing decoherence in quantum computer memory.
\newblock {\em Phys. Rev. A}, 52:R2493--R2496.

\bibitem[Steane, 1996]{PhysRevLett.77.793}
Steane, A.~M. (1996).
\newblock Error correcting codes in quantum theory.
\newblock {\em Phys. Rev. Lett.}, 77:793--797.

\bibitem[Wang et~al., 2017]{wang2017adjustable}
Wang, H., Zhou, J., Wu, J., and Yi, X. (2017).
\newblock Adjustable self-loop on discrete-time quantum walk and its
  application in spatial search.
\newblock {\em arXiv preprint arXiv:1707.00601}.

\bibitem[Wong, 2015]{1751-8121-48-43-435304}
Wong, T.~G. (2015).
\newblock Grover search with lackadaisical quantum walks.
\newblock {\em Journal of Physics A: Mathematical and Theoretical},
  48(43):435304.

\end{thebibliography}

\end{document}